\documentclass[10pt]{llncs}

\usepackage{amsmath}
\usepackage{verbatim}

\begin{document}

\title{On the hardness of distance oracle for sparse graph}
\author{Hagai Cohen \and Ely Porat\thanks{This work was supported by BSF and ISF}}
\institute{Department of Computer Science, Bar-Ilan University, 52900 Ramat-Gan, Israel
\email{\{cohenh5,porately\}@cs.biu.ac.il}}

\maketitle

\begin{abstract}
In this paper we show that \emph{set-intersection} is harder than
distance oracle on sparse graphs.
Given a collection of total size $n$ which consists of $m$ sets drawn from universe $U$,
the \emph{set-intersection} problem is to build a data structure
which can answer whether two sets have any intersection.
A distance oracle is a data structure which can answer distance queries on a given graph.
We show that if one can build distance oracle for sparse graph $G = (V, E)$,
which requires $s(|V|, |E|)$ space and answers
a $(2-\epsilon, c)$-approximate distance query in time $t(|V|, |E|)$
where $2-\epsilon$ is a multiplicative error and $c$ is a constant additive error,
then, \emph{set-intersection} can be solved in $t(m+|U|, n)$ time using $s(m+|U|, n)$ space.
\end{abstract}

\section{Introduction and Related Work}
Let $G = (V,E)$ be a graph.
The all-pairs shortest paths problem (APSP) requires
to construct a data structure for a given graph $G$
so that the exact distance between every two vertices on that graph can be retrieved efficiently.
This problem is one of the most fundamental graph problems of computer science.
Despite the importance of this problem,
there is still no efficient solution for it using less than $O(|V|^2)$ space.
When the graph is dense, i.e., when $|E|=O(|V|^2)$, this space is not much.
But, for sparse graphs where $|E|=O(|V|)$ this is extremely a lot of space.

Thorup and Zwick \cite{mikkel01} explored an alternative for the APSP problem.
They introduced a solution for the approximate distance oracle,
which is a data structure that answers approximate distance queries in a graph.
They achieved that, for any integer $k \geq 1$, an undirected weighted graph with $n$ vertices and $m$ edges can be preprocessed in expected $O(kmn^{1/k})$ time to construct a data structure of size $O(kn^{1+1/k})$ that can answer any $(2k-1)$-approximate distance query in $O(k)$ time.
This means that the distance oracle answers distance queries with multiplicative error of $2k-1$.

In this paper we show by a reduction from the set intersection problem
that it is hard to build a
$(2-\epsilon, c)$-approximate distance oracle
where $2-\epsilon$ is a multiplicative error and $c$ is a constant additive error.

In the set intersection problem, we are given a collection of sets which we can preprocess.
Then, given two sets we need to answer quickly whether there is any intersection between the sets.
This is a common problem in many fields,
especially in retrieval algorithms and search engines.
The formal definition of the problem is as follows:
\begin{definition}
Let $D$ be a database consisting of a collection of $m$ sets drawn from universe $U$,
$S_1, \ldots, S_m \subseteq U$.
Denote $n$ to be the input size, i.e., $n = \sum_{i=1}^m |S_i|$.
The \emph{set intersection} problem is to build a data structure
that given a query of two indices $i,j \leq m$,
can answer if sets $S_i$ and $S_j$ have any intersection.
\end{definition}

Cohen and Porat \cite{hagai10} showed how the \emph{set intersection} problem
can be solved in $O(\sqrt{n})$ query time using $O(n)$ space.
Their solution is based on dividing the sets in the database $D$ to large and non-large sets,
where they define a large set to be a set which has more than $\sqrt{n}$ elements.
They construct a set intersection matrix for the large sets in $D$,
which is a matrix saving for each pair of sets if there is any intersection between them.
They showed that the number of large sets is at most $\sqrt{n}$,
thus, this matrix costs $\sqrt{n} \times \sqrt{n} = O(n)$ bits space.
Moreover, for each set in $D$ they store a static hash table to retrieve in $O(1)$ time
if an element belongs to that set or not.

Given a query consisting of two indices $i,j$,
if both $S_i$ and $S_j$ are large sets,
the answer can be retrieved from the set intersection matrix in $O(1)$ time.
Otherwise, one of the sets is a non-large set, i.e., it has less than $\sqrt{n}$ elements.
On this case, the answer can be retrieved by going over all the elements of the smaller set,
checking for each one of them if it belongs to the other set in $O(1)$ time.
Because non-large sets have at most $O(\sqrt{n})$ elements, this takes at most $O(\sqrt{n})$ time.

This solution can be easily extended to a tunable solution.
If we define a large set to be a set with more than $t$ elements,
the number of large sets can be at most $\frac{n}{t}$ sets.
Thus, the set intersection matrix costs $O(\frac{n^2}{t^2})$ space.
Hence, this problem can be answered in $O(t)$ query time
using $O(\frac{n^2}{t^2})$ space.

In this paper we show a reduction from the \emph{set intersection} problem
to distance oracle on sparse graphs.
In Sect.~\ref{sec:Set Intersection Reduction} we show that if one can build a distance oracle
using $s(|V|, |E|)$ space with $t(|V|, |E|)$ query time,
which answers $(2-\epsilon)$-approximate distance queries,
the \emph{set intersection} problem can be solved
in $t(m+|U|, n)$ query time using $s(m+|U|, n)$ space.
In Sect.~\ref{sec:Distance Oracle with Constant Additive Error} we extend the reduction
to a $(2-\epsilon)$-approximate distance oracle with constant additive error.

\section{Set Intersection Reduction}\label{sec:Set Intersection Reduction}

In the next theorem we claim that if one can build a distance oracle
that answers $(2-\epsilon)$-approximate distance queries,
the \emph{set intersection} problem can be solved.

\begin{theorem}\label{thm:basic reduction}
Let $G = (V, E)$ be a sparse graph.
Given a distance oracle that answers $(2-\epsilon)$-approximate distance queries
using $s(|V|, |E|)$ space with $t(|V|, |E|)$ query time,
we can solve the set intersection problem using $s(m+|U|, n)$ space with $t(m+|U|, n)$ query time.
\end{theorem}

\begin{proof}
For the set intersection problem we are
given a database $D$ consisting of $m$ sets drawn from universe $U$,
$S_1, \ldots, S_n \subseteq U$.
We denote $n$ to be the input size, i.e., $n = \sum_{i=1}^m S_i$.

We construct a bipartite graph with two disjoint sets of vertices:
$V_1$ with vertices for each set in $D$ and $V_2$ with vertices for each element in $U$.
Hence, $|V_1| = m$ and $|V_2| = |U|$.
The edges between $V_1$ and $V_2$ are simple,
if an element $e \in U$ belongs to a set $s$ then there is an edge between the corresponding vertices in the bipartite graph.
Because this graph is a bipartite graph it is simple that the distance between each two vertices must be even.
We can see that if two sets have any intersection between them,
the distance between the corresponding vertices is $2$.
The number of edges on this graph is bounded by $n$.
We construct a distance oracle for this graph which answers $(2-\epsilon)$-approximate distance queries
using $s(m+|U|, n)$ space.

Given two sets $S_i,S_j$ we would like to calculate if there is any intersection between them.
To answer that we retrieve the approximate distance between
the corresponding vertices of $S_i$ and $S_j$ in the bipartite graph.
Because both the vertices are in $V_1$
if the approximate distance is less than $4-\epsilon$,
the exact distance must be $2$ because the distance must be even.
This means that there is an element $e$ that has an edge to either $v_i$ and $v_j$,
therefore, there is an intersection between $S_i$ and $S_j$.
Otherwise, there is no such an element,
hence, there is no intersection between the sets.
Therefore, we can answer the set intersection problem in $t(m+|U|, n)$ query time
using $s(m+|U|, n)$ space.
\qed
\end{proof}

\section{Distance Oracle with Constant Additive Error}\label{sec:Distance Oracle with Constant Additive Error}
In this section we extend the reduction to $(2-\epsilon)$-approximate distance oracle with constant additive error.
We prove that set intersection is harder than approximate distance oracle even for distance oracle with constant additive error.
We denote a distance oracle with multiplicative error $d$
and additive error $c$ as $(d, c)$-approximation distance oracle.

\begin{theorem}
Let $G = (V, E)$ be a sparse graph.
Given a distance oracle that answers $(2-\epsilon, c)$-approximate distance queries
using $s(|V|, |E|)$ space with $t(|V|, |E|)$ query time,
we can solve the set intersection problem using $s(m+|U|, n)$ space with $t(m+|U|, n)$ query time.
\end{theorem}

\begin{proof}
We build a bipartite graph as in the proof of Theorem~\ref{thm:basic reduction}.
But now each edge between $V_1$ and $V_2$ will be a path of $\frac{2}{\epsilon}(c-1)$ vertices.
This adds a constant number of vertices and edges, hence,
the space cost of a distance oracle for the bipartite graph is still $s(m+|U|, n)$ space.

Given two sets $S_i, S_j$ we would like to calculate if there is any intersection between them.
To answer that we retrieve the approximate distance between
the corresponding vertices of $S_i$ and $S_j$ in the bipartite graph.

If the approximated distance is less than $\frac{4c}{\epsilon}$,
because the distance oracle is $(2-\epsilon, c)$-approximate distance oracle,
it means that the exact distance has to be less than
$\frac{4c}{\epsilon} \times (2-\epsilon) + c = \frac{8c}{\epsilon}-4c+c$.
The minimal distance between vertex in $V_1$ and vertex in $V_2$ is $\frac{2}{\epsilon}(c-1)$,
hence,
the distance is exactly $\frac{4}{\epsilon}(c-1)$ and therefore there is an intersection.
If the approximated distance is greater than $\frac{4c}{\epsilon}$,
there would be no intersection because the distance is too high.

By that we solved the set intersection problem in $t(m+|U|, n)$ time 
using $s(m+|U|, n)$ space.

\qed
\end{proof}

\section{Conclusions}\label{sec:Conclusions}
In this paper we showed that set intersection is harder than distance oracle on sparse graphs.
We showed how the set intersection problem can be solved using $(2-\epsilon)$-approximate distance oracle with constant additive error.

\bibliographystyle{splncs}
\bibliography{intersectionEmptyBib}

\end{document}